\documentclass[10pt,letterpaper]{article}

\usepackage{algorithm,algorithmic}
\usepackage{authblk}
\usepackage{amsmath}
\usepackage{amssymb}
\usepackage{amsthm}
\usepackage[T1]{fontenc}
\usepackage{geometry}
\usepackage{lmodern}
\usepackage{multirow}
\usepackage{footnote}

\newcommand{\ie}{{\em i.e., }} 
\newcommand{\etal}{{\em et al.}}

\newcommand{\suchthat}{\mid} 

\newcommand{\outputs}{\pi_{\mathit{out}}}

\newcommand{\popt}{P_\mathrm{TO}} 
\newcommand{\bfpopt}{\boldsymbol{\popt}} 

\newcommand{\pep}{P_\mathrm{EP}} 
\newcommand{\pcd}{P_\mathrm{CD}}

\newcommand{\rs}{\mathbf{\Gamma}}

\newcommand{\ifthen}[2]{\textbf{if} #1 
\textbf{then} #2 \textbf{endif}} 
 
\newcommand{\fordo}[2]{\textbf{for} #1 \textbf{do} #2}
\newcommand{\vblank}{\vspace{0.2cm}}

\newcommand{\rmax}{r_{\mathrm{max}}}
\newcommand{\rmid}{r_{\mathrm{mid}}}
\newcommand{\smax}{s_{\mathrm{max}}}
\newcommand{\bmax}{b_{\mathrm{max}}}
\newcommand{\tr}{1}
\newcommand{\fl}{0}
\newcommand{\para}{\tau}

\newcommand{\checking}{\mathit{CH}}
\newcommand{\elect}{\mathit{EL}}

\newcommand{\cand}{A}
\newcommand{\timer}{B}
\newcommand{\vl}{V_{L}}
\newcommand{\vf}{V_{F}}
\newcommand{\vcheck}{V_{\checking}}
\newcommand{\velect}{V_{\elect}}
\newcommand{\vclarge}{V_{\checking \ge}}
\newcommand{\vcsmall}{V_{\checking <}}
\newcommand{\vtimer}{V_{\timer}}
\newcommand{\vcand}{V_{\cand}}
\newcommand{\vdone}{V_{\mathrm{done}}}
\newcommand{\vundone}{V_{\mathrm{undone}}}

\newcommand{\leader}{\mathtt{leader}} 
\newcommand{\phase}{\mathtt{phase}}
\newcommand{\mode}{\mathtt{mode}}
\newcommand{\rtimer}{\mathtt{timer}_R} 
\newcommand{\detect}{\mathtt{detect}}
\newcommand{\level}{\mathtt{level}}
\newcommand{\done}{\mathtt{done}}
\newcommand{\btimer}{\mathtt{timer}_{\timer}}

\newcommand{\quick}{\mathit{QE}}
\newcommand{\bfquick}{\boldsymbol{\quick}} 

\newcommand{\toelect}{\mathit{GoToElection}}

\newcommand{\ex}{\mathbf{E}} 

\newcommand{\eht}[3]{\mathrm{EIH}_{#1}(#2,#3)} 
\newcommand{\ect}[3]{\mathrm{EIC}_{#1}(#2,#3)} 
\newcommand{\sle}{\mathit{LE}}

\newcommand{\cala}{\mathcal{A}}
\newcommand{\call}{\mathcal{C}_{\mathrm{all}}} 
\newcommand{\ccheck}{\mathcal{A}_{\checking}}
\newcommand{\celect}{\mathcal{A}_{\elect}}
\newcommand{\creset}{\mathcal{C}_{\mathrm{reset}}}
\newcommand{\cclarge}{\mathcal{A}_{\checking \ge}}
\newcommand{\ccsmall}{\mathcal{A}_{\checking <}}
\newcommand{\celarge}{\mathcal{C}_{\elect \ge}}
\newcommand{\safe}{\mathcal{S}}

\newcommand{\lagent}{v_l}

\newcommand{\nva}{n_{\cand}}
\newcommand{\nvm}{n_{M}}
\newcommand{\pinc}{p_{\mathrm{inc}}}
\newcommand{\pdec}{p_{\mathrm{dec}}}

\newcommand{\logn}{\lceil \log n \rceil}

\newcommand{\cell}[2]{\begin{tabular}{#1}#2\end{tabular}}

\newtheorem{theorem}{Theorem}
\newtheorem{lemma}{Lemma}
\theoremstyle{definition}
\newtheorem{definition}{Definition}

\title{Time-optimal Loosely-stabilizing Leader Election in Population Protocols}
\date{}
\author[1]{Yuichi Sudo\thanks{Corresponding author:y-sudou[at]ist.osaka-u.ac.jp}}
\author[2]{Ryota Eguchi}
\author[2]{Taisuke Izumi}
\author[1]{Toshimitsu Masuzawa}

\affil[1]{Graduate School of Information Science and Technology, Osaka University, Japan}
\affil[2]{Graduate School of Engineering, Nagoya Institute of Technology, Japan}

\begin{document}
\maketitle
\begin{abstract}
\normalsize
We consider the leader election problem in population protocol models. In pragmatic settings of population protocols, self-stabilization is a highly desired feature owing to its fault resilience and the benefit of initialization freedom.
However, the design of self-stabilizing leader election is possible only under a strong assumption
(\ie the knowledge of the \emph{exact} size of a network) and rich computational 
resource (\ie the number of states). Loose-stabilization, introduced by Sudo \etal~[Theoretical Computer Science, 2012],
is a promising relaxed concept of self-stabilization to address the aforementioned issue. 
Loose-stabilization guarantees that starting from any configuration, the network will reach a safe configuration where a single leader exists within a short time, and thereafter it will maintain the single leader for a long time, but not forever. The main contribution of the paper is a time-optimal loosely-stabilizing leader election protocol. While the shortest convergence time achieved so far in loosely-stabilizing leader election is $O(\log^3 n)$ parallel time,
the proposed protocol with design parameter $\para \ge 1$ attains $O(\para \log n)$ parallel convergence time and
$\Omega(n^{\tau})$ parallel holding time (\ie the length of the period keeping the unique leader), both in expectation.
This protocol is time-optimal in the sense of both the convergence and holding times in expectation because
any loosely-stabilizing leader election protocol with the same length of the holding time
is known to require $\Omega(\tau \log n)$ parallel time.
\end{abstract}

\section{Introduction}
\label{sec:intro}
We consider the \emph{population protocol} 
(PP) model \cite{original} in this paper. 
A network called the \emph{population} consists 
of $n$ automata called \emph{agents}. 
Pairs of agents execute \emph{interactions} 
(\ie pairwise communication) 
by which they update their states. 
These interactions are opportunistic, 
that is, they are unknown and unpredictable 
(or only predictable with probability). 
Agents are strongly anonymous: 
they do not have identifiers 
and cannot distinguish 
neighbors with the same state. 
As with the majority of studies on population 
protocols 
\cite{original,fast,AG15,BCER17,AAG18,BKKO18,GS18,GSU18,Sud+20,SOI+19+}, 
we assume 
that exactly one pair of agents
is selected to have an interaction
uniformly at random from all $\binom{n}{2}$ pairs
at each step. 
In the PP model, 
time complexity such as expected convergence time 
is usually evaluated in 
\emph{parallel time},
that is, the number of 
steps divided by $n$ (\ie the number of agents). 
This is a natural measure of time because 
in practice, 
interactions typically occur in parallel in 
the population. 
For the remainder of this section,
we presume parallel time
when we discuss time complexity.

In this paper, we focus on the problem of 
self-stabilizing leader election (SS-LE).
This problem requires 
that (i) starting from any configuration, 
a population reaches a safe configuration in 
which exactly one leader exists; 
and thereafter, (ii) it keeps this leader forever. 
These requirements guarantee 
tolerance against finitely many transient faults. 
Since many protocols (self-stabilizing or non-self-stabilizing) 
in the literature assume a unique leader 
\cite{original,jikoantei,fast}, 
SS-LE is key to improving fault-tolerance 
of the PP model itself. 
However, it is known that no protocol can solve SS-LE 
unless every agent in the population 
knows the \emph{exact size} $n$ of the population 
\cite{jikoantei,SIW12} \footnotemark{}.
\footnotetext{ 
Strictly speaking, they prove a slightly weaker
impossibility. However, we can prove this impossibility
based on almost the same technique:
a simple partitioning argument.
See \cite{Sud+20} for details (page 618, footnote). 
} 
Under this strong assumption
(\ie all agents know exact $n$),
several SS-LE protocols have been presented in the literature.
Cai \etal~\cite{SIW12} gave
the first SS-LE protocol under this assumption,
which elects the unique leader within $O(n^2)$ time
starting from any configuration.
Recently, Burman \etal~\cite{Bur+19} gave three SS-LE protocols,
which improve the convergence time at the cost of space complexity, that is, the number of states per agent
(Table \ref{table:le}). For example, one of their protocols converges in $O(n)$ time but uses a super-exponential number of states.

\begin{savenotes}
\begin{table}[t] 
\center 
\caption{Self/Loosely-stabilizing leader 
election in the PP model 
(shown in parallel time)} 
\label{table:le} 
\begin{tabular}[t]{c c c c c c c} 
\hline 
& Type & \cell{c}{Knowledge $N$}  
&\cell{c}{Convergence\\ time}& 
\cell{c}{Holding\\ time}& \cell{c}{\#states}
& \cell{c}{Design \\ parameter}\\ 
\hline 
\cite{SIW12} 
&SS-LE&$N=n$& $O(n^2)$ &$\infty$ 
&$n$ &-\\ 
\cite{Bur+19}
&SS-LE&$N=n$& $O(n \log n)$ &$\infty$ 
&$O(n)$ &- \\ 
\cite{Bur+19}
&SS-LE&$N=n$& $O(n)$ &$\infty$ 
&$n^{O(n)}$ &- \\ 
\cite{Bur+19}
&SS-LE&$N=n$& $O(\log n)$ &$\infty$ 
&$\infty$ &-\\ 
\hline 
\cite{kanjiko} 
&LS-LE&
$n \le N=O(n)$
&$O(n)$~
\footnote{The convergence time of this protocol was proven
to be $O(n \log n)$ in \cite{kanjiko}.
Later, it was found to be $O(n)$ according to
Lemma 1 in \cite{ADK+17}.}
&$\Omega(e^n)$ 
&$O(n)$ &-\\ 
\cite{kakai} 
&LS-LE&
$n \le N=O(n)$
&$O(n)$ &$\Omega(e^{n})$ &$O(n)$
&-\\ 
\cite{Sud+20}
&LS-LE& 
$n \le N=\mathit{poly}(n)$
&$O(\para \log^3 n)$ &$\Omega(n^{\para})$ 
&$O(\para^2 \log^5 n)$
& any $\para \ge 10$\\ 
ours
&LS-LE& 
$n \le N=\mathit{poly}(n)$
&$O(\para\log n)$ &$\Omega(n^{\para})$ 
&$O(\para \log n)$
& any $\para \ge 1$ \\ 
\hline 
\end{tabular} 
\end{table} 
\end{savenotes}

We can discard the assumption of exact-$n$-knowledge
by slightly relaxing the requirement of self-stabilization,
that is, by taking an approach called \emph{loose-stabilization}.
Loose-stabilization guarantees 
that
the population reaches a safe 
configuration 
within a relatively short time
starting from any initial configuration; 
after that, the specification of the problem 
(such as having a unique leader in the leader election) 
must be sustained for a sufficiently long 
time, though not necessarily forever. 
Sudo \etal~\cite{kanjiko}
gave a loosely-stabilizing leader election (LS-LE) 
protocol by assuming that every agent knows a common 
\emph{upper bound} $N$ of $n$. 
Their protocol is not self-stabilizing;
however, it is practically equivalent 
to an SS-LE protocol because
it maintains the unique leader for an
exponentially long time 
after reaching a safe configuration.
Further, it converges in a safe configuration within $O(N)$ time
starting from any configuration.
Hence, the convergence time is $O(n)$ if 
we have a good upper bound $N=O(n)$.
In practice, the knowledge of $N$ 
is a much weaker assumption than exact-$n$-knowledge;
the protocol works correctly 
even if we consider a large overestimation of $n$, 
such as $N = 100n$. 
%
%
Recently, Sudo \etal~\cite{Sud+20} gave 
an LS-LE protocol with poly-logarithmic convergence time,
which has a design parameter $\para~(\ge 10)$
controlling the convergence and holding times.
Given an upper bound $N \ge n$ such that $N=O(n^c)$
for some constant $c$,
their protocol reaches a safe configuration
within $O(\para \log^3 n)$ time,
and thereafter, it keeps the single leader
for $\Omega(n^{\para})$ time, both in expectation.


Izumi \cite{kakai} provided a lower bound
on the convergence time of an LS-LE protocol,
given that it keeps the unique leader 
for an exponentially long time 
after reaching a safe configuration.
Sudo \etal~\cite{Sud+20} generalized this lower bound as follows:
if 
the expected holding time of an LS-LE protocol is $\beta/n$,
its expected convergence time must be $\Omega(\log \beta)$.
Therefore, 
we have a gap of $\log^2 n$ factor between
this lower bound 
and the upper bound given by Sudo \etal~\cite{Sud+20}
when we require an expected holding time of $\Omega(n^{\para})$:
the former is $\Omega(\para \log n)$
and the latter is $O(\para \log^3 n)$.


\subsection{Our Contribution} 
\label{sec:contribution} 
We remove the above-mentioned gap in this paper.
That is, we develop an LS-LE protocol
whose expected convergence time
is $O(\para \log n)$  
and expected holding time is $\Omega(n^{\para})$,
where $\para \ge 1$ is the design parameter of the protocol.
Interestingly, this convergence time is optimal for any length of holding time.
For $\para \ge 1$, we have no asymptotic gap
between this convergence time and
the lower bound given by Sudo \etal~\cite{Sud+20}.
Even if a holding time of $o(n)$ is sufficient,
the expected convergence time of our protocol with $\para=1$ remains optimal.
This is because every LS-LE protocol requires $\Omega(\log n)$ time to reach a safe configuration regardless of the length of its holding time. Consider an execution of any LS-LE protocol starting from a configuration where all agents are leaders.
Then, $n-1$ agents must have at least one interaction before
electing the unique leader. 
However, a simple analysis on the famous \emph{coupon collector's  problem}
yields that this requires $\Omega(\log n)$ time
(\ie $\Omega(n \log n)$ steps) in expectation. 
In addition to time-optimality, 
the proposed protocol has a small space complexity:
The number of states per agent is $O(\para \log n)$,
which is much smaller than $O(\para^2 \log^5 n)$ in \cite{Sud+20}.

The proposed protocol also shows how useful
loose-stabilization is in the PP model. 
When we set $\para = 100$, its expected convergence time
is $O(\log n)$, and the expected holding time is $\Omega(n^{100})$,
practically forever. This protocol needs only the knowledge of $N$
such that $n \le N = O(n^c)$ holds for some constant $c$.
Under self-stabilization,
if we require the same convergence time, the unique solution in the literature \cite{Bur+19} uses the infinite space of each agent and requires a much stronger assumption, \ie the knowledge of exact $n$.


\subsection{Related Work} 
Leader election has been extensively 
studied in the PP model. 
When we design non-self-stabilizing 
protocols, we can assume that 
all agents are in a specific state at the 
initial configuration. 
Leader election is then achieved 
by employing a simple protocol \cite{original}. 
In this protocol, all agents are initially 
leaders, 
and we have only one transition rule: 
when two leaders meet, one of them
becomes a follower (\ie a non-leader). 
This simple protocol elects a unique leader 
in linear time
and uses only two states at each agent. 
This protocol is time-optimal:
Doty and Soloveichik \cite{DS18} showed that 
any constant space protocol requires linear 
time to elect a unique leader. 
In a breakthrough result, 
Alistarh and Gelashvili \cite{AG15} 
designed a (non-self-stabilizing) leader election protocol 
that converges in $O(\log^3n)$ 
parallel time and uses $O(\log^3 n)$ states at each 
agent. 
Thereafter, a number of papers have been devoted to 
fast leader election \cite{BCER17, AAG18, 
BKKO18, GS18, GSU18, SOI+18, BGK20}.
G{\k{a}}sieniec and Staehowiak \cite{GSU18} 
gave an algorithm that converges in 
$O(\log n \log \log n)$ time
and uses a surprisingly small number of states:
only $O(\log \log n)$ states per agent.
This is space-optimal because it is known that
every leader election protocol with poly-logarithmic time
uses $\Omega(\log \log n)$ states \cite{AAE+17}. 
Sudo \etal~\cite{SOI+18} gave
a protocol that elects a unique leader within 
$O(\log n)$ time
and uses $O(\log n)$ states per agent. 
This is time-optimal because
any leader election protocol 
requires $\Omega(\log n)$ time
even if it uses an arbitrarily large number of states 
and the agents know the exact size of the population \cite{SM19}.
\footnote{This lower bound is not trivial:
it does not immediately follows from a simple analysis of the coupon collector's problem
because unlike SS-LE/LS-LE setting, we can now specify an initial configuration
such that all agents are followers.}
These two protocols were the state-of-the-art 
until Berenbrink \etal~\cite{BGK20}
gave a time and space optimal protocol very recently.

LS-LE protocols 
are also presented for a population where 
some pairs of agents may not have 
interactions, 
\ie the interaction graph is not complete \cite{SOK+18,without}. 

\section{Preliminaries} 
\label{sec:preliminaries} 

\subsection{Model}
\label{sec:model}
We denote the set of integers 
$\{z \in \mathbb{N}\suchthat x \le z \le y\}$ 
by $[x,y]$.
The omitted bases of logarithms are 2.

A \emph{population}  
is the set $V$ of $n$ {\em agents}
(\ie $|V|$ = n) that changes their states by pairwise interactions.
Every pair of agents $(u,v) \in V \times V \setminus \{(w,w) \mid w \in V\}$
can interact with each other.
A \emph{protocol} $P$ on the population is defined by a 4-tuple 
$P = (Q,Y,T,\outputs)$
consisting of
a finite set $Q$ of states, 
a finite set $Y$ of output symbols, 
a transition function 
$T: Q \times Q \to Q \times Q$, 
and an output function $\outputs : Q \to Y$. 
When two agents interact, 
$T$ determines their next states 
based on their current states. 
The output function $\outputs$ maps the current local state $q \in Q$ to a value in the output domain $\outputs(q) \in Y$.
The state of each agent including the current output are often described as a set 
of local variables.
Throughout this paper, we use the notation $v.x$ to denote the 
value of a variable $x$ managed by agent $v$.

We assume that all agents have a common knowledge $N$ on $n$ such that $n \le N = O(n^c)$ holds for some constant $c$,
which is equivalent to the assumption that the agents have 
a constant-factor approximation $m$ of $\log n$, \ie  $\alpha \log n \geq m \geq \log n$
for some constant $\alpha \geq 1$\footnote{In this sense, any protocol $P$ should be parametric (with respect to $m$) such as $P_m = (Q_m, Y_m, T_m, \pi_{\mathrm{out}, m})$ strictly. In this paper, we do not explicitly state parameter $m$ of $P$ for simplicity.}.

A \emph{configuration} is a mapping $C : V \to Q$ that specifies 
the states of all agents. 
Given a protocol $P$ on $n$ agents, the set of all possible configurations for $P$
is denoted by $\call(P)$. 
We say that a configuration $C$ changes to 
$C'$ by an interaction 
$e = (u,v)$, 
denoted by $C \stackrel{P,e}{\to} C'$, 
if 
$(C'(u),C'(v))=T(C(u),C(v))$ 
and $C'(w) = C(w)$ 
for all $w \in V \setminus \{u,v\}$. 
Then $u$ and $v$ are respectively called the \emph{initiator} 
and the \emph{responder} of $e$. Given an interaction $e$, we say that agent $v \in V$ 
\emph{participates} in $e$ if $v$ is either the initiator or the 
responder of $e$.

We assume the \emph{uniformly random scheduler} $\rs$, which
selects two agents to interact at each step uniformly at random from all pairs of agents.
Specifically, $\rs=\Gamma_0, \Gamma_1,\dots$
where each $\Gamma_t \in E$ is a random variable such 
that 
$\Pr(\Gamma_t = (u,v)) =\frac{1}{n(n-1)}$ 
for any $t \ge 0$ and any distinct $u,v \in 
V$. 
Given an initial configuration $C_0 \in \call(P)$,
the \emph{execution} of protocol $P$ under the uniformly random scheduler $\rs$
is defined as $\Xi_{P}(C_0,\rs) = C_0,C_1,\dots$ 
where $C_t \stackrel{P,\Gamma_t}{\to} C_{t+1}$ holds for all $t \ge 0$. 
Note that each $C_i$ is also a random variable.

\subsection{Loosely-Stabilizing Leader Election}
In leader election protocols, every agent is equipped with an output variable 
$\leader \in \{\fl, \tr\}$, which indicates whether the agent is a leader. That is, if $v.\leader \{\fl,\tr\}$ holds, $v$ is a leader, and a follower otherwise.
A configuration $C$ is called {\em correct with leader} $v \in V$ 
if $v$ outputs $\tr$ and all other agents output $\fl$. 
Given any configuration $C$, 
we define $\eht{P}{C}{\sle}$ 
as the expected length of the longest prefix of $\Xi_P(C,\rs)$, where any configuration is correct with 
a common leader $v \in V$. Note that $\eht{P}{C}{\sle}=0$ holds if 
a configuration $C$ is not correct. For any configuration $C$ and any subset $\mathcal{S} \subseteq \call(P)$ 
of configurations, we also define 
$\ect{P}{C}{\mathcal{S}}$ as the expected length of the longest prefix of $\Xi_P(C,\rs)$, where any 
configuration is not in $\mathcal{S}$

\begin{definition}[Loosely-stabilizing leader election \cite{kanjiko}] 
Let $\alpha$ and $\beta$ be positive real numbers. 
Protocol $P(Q,Y,T,\outputs)$ is 
an $(\alpha,\beta)$-loosely-stabilizing 
leader election protocol 
if there exists a set $\mathcal{S}$ of 
configurations 
satisfying the two inequalities
\begin{align*} 
\max_{C \in \mathcal{C}_\mathrm{all}(P)} 
\ect{P}{C}{\mathcal{S}} \le \alpha 
~~~\text{and}~~~ 
\min_{C \in \mathcal{S}} \eht{P}{C}{\sle} \ge 
\beta. 
\end{align*} 
\end{definition} 

We call $\mathcal{S}$ defined by the definition above the set of \emph{safe} configurations of $P$.
Note that the condition $\beta > 0$ guarantees the correctness ({\ie} uniqueness of leader) of configurations in $\mathcal{S}$. 
In terms of parallel time, 
an $(\alpha,\beta)$-loosely-stabilizing
leader election 
protocol $P$ reaches a safe configuration 
within $\alpha/n$ parallel time in 
expectation, and it keeps the elected leader during the following $\beta/n$ 
parallel time in expectation.  
We call $\alpha/n$ and $\beta/n$ 
\emph{the expected convergence time} 
and \emph{the expected holding time} of 
$P$, respectively. 

\section{Toolbox}

\subsection{Epidemic}
\label{sec:epidemic}
The protocol \emph{epidemic} \cite{fast}, denoted by $\pep$,
is often used to propagate the maximum value of a variable to the whole population, which is defined 
as:
(i) each agent has only one variable $x$,
and (ii) when two agents $u$ and $v$ interact,
they substitute $\max(u.x,v.x)$ for their variables (\ie $u.x$ and $v.x$).
Then, we have the following lemma.

\begin{lemma}[\cite{fast} \footnotemark{}]
\label{lemma:epidemic}
Let $k$ be any non-negative integer, $D_0 \in \call(\pep)$ be any configuration of $\pep$, and $l=\max_{v \in V} v.x$ in configuration $D_0$. The execution $\Xi_{\pep}(D_0,\rs)$ reaches the configuration such that $u.x = l$ holds for 
any $u \in V$ within $O(kn \log n)$ steps with probability $1-O(n^{-k})$.
\end{lemma}

\footnotetext{
While the original protocol by Angluin \etal~\cite{fast} is an one-way version of $\pep$ (\ie higher value is propagated only from an initiator to a responder), there is no difference on asymptotic propagation time between them (Lemma 8 in \cite{kanjiko}).
}

\subsection{Countdown with Higher Value Propagation}
\label{sec:countdown}
The protocol of
\emph{counting down with higher value propagation} (CHVP)
\cite{kanjiko} is a useful technique to design loosely-stabilizing protocols, particularly for detecting 
the absence of a leader. It is defined as the following protocol $\pcd$: each agent has only one variable $y$,
and when two agents $u$ and $v$ interact,
they substitute $\max(u.y-1,v.y-1,0)$ for their $y$. We have the following two lemmas.
\begin{lemma}[Lemma 1 in \cite{ADK+17}\footnote{
Precisely, $k$ is assumed to be a constant in the original lemma, but the same proof applies in the case that $k$ depends on $n$.}]
\label{lemma:countdown} 
Let $l_1$ and $l_2$ be any two integers
such that $l_1 > l_2 \ge 0$,
$k$ be any non-negative integer,
and $D_0 \in \call(\pcd)$ be any configuration of $\pcd$ such that $l_1=\max_{v \in V} v.y$ holds. 
The execution $\Xi_{P}(D_0,\rs)$ reaches a configuration satisfying $\max_{v \in V} v.y \le l_2$
within $O(n(l_1-l_2 + k\log n))$ steps with probability $1-O(n^{-k})$.
\end{lemma}
\begin{lemma}[Lemma 5 in \cite{Sud+20}\footnote{
We obtain this lemma
by substituting $d=k+3$, $d'=k+3$, $d''=6$,
and $t = \lceil kn \ln n \rceil$ for the first inequality
in Lemma 5 in \cite{Sud+20}.
}]
\label{lemma:prop} 
Let $D_0 \in \call(\pcd)$ be any configuration and $l$ be the integer that satisfies
$l=\max_{v \in V} v.y$ at $D_0$. There exists a constant $c$ such that
$\Xi_{\pcd}(D_0,\rs)$ reaches a configuration satisfying $\min_{v \in V} v.y \ge l - ck \lg n$ within $O(kn\log n)$ steps with probability $1-O(n^{-k})$.
\end{lemma}

\subsection{Lottery Game and Quick Elimination}
\label{sec:game}
The \emph{lottery game}, originally introduced by Alistarh \etal~\cite{AAE+17} as 
a part of their leader election protocol, is a probabilistic
process of filtering leaders. An abstract form of the lottery game is stated 
as follows: Let $V'$ be the set of leaders. Every leader $v \in V'$ makes independent fair coin flips until it observes tail for the first time. Then, the number of observed heads $s_v$ (called the \emph{level} of $v$) is propagated to other leaders. The agent identifying another agent with a higher level drops out as a loser. 

There are a few implementations of the lottery game in population 
protocol models. Alistarh \etal~\cite{AAE+17} and Sudo \etal~\cite{SOI+19+} develop (non-loosely-stabilizing)
leader election protocols, based
on their own implementations and analyses for this game. In this paper, we adopt the
implementation shown in~\cite{SOI+19+}, called \emph{quick elimination} (QE). 
The pseudocode of QE is given in
Algorithm~\ref{al:quick}, which describes the state transition when two agents $a_0$ 
and $a_1$ interact.
Since the propagation of level values is easily implemented 
by the epidemic, the main non-trivial point is
how to synthesize coin flips using the randomness of the scheduler.
The implementation QE simply utilizes the asymmetry of interactions. That is, if 
$v$ joins an interaction as the initiator, it receives head as the result of 
its coin flip, and receives tail if it joins as the responder. Each agent maintains 
two variables, $\done$ and $\level$, in addition to an output variable $\leader$. The 
flag $\done \in \{\fl, \tr\} $ implies whether the agent is still in the decision 
of its level (\ie, it continues (synthetic) coin flips during $\done = \fl$).
Starting from the state with $\done = \fl$ and $\level = 0$,
the agent $v$ with $v.\leader = \tr$ first decides its level:
it increments $v.\level$ every time it observes head,
and it stops incrementation and sets $\done$ to $\tr$ when it observes tail for the first time.
Agents that have decided their levels perform the epidemic to share the maximum level (lines 5-7). If an agent sees a higher level, it becomes a follower (line 6).

While the lottery game was used as a scheme to eliminate leaders in the past literature, we rather see it as a Monte Carlo protocol for leader election, \ie, we focus on the probability that exactly one player wins (or survives as a leader). The following lemma is the key ingredient of our protocol, which is simple but a new observation that has not been addressed so far.

\begin{algorithm}[t]
\caption{$\bfquick()$}
\label{al:quick}
\begin{algorithmic}[1]  
 \IF{$\exists i \in \{0,1\} : a_i.\done = \fl \wedge a_i.\leader = \tr$}
 \STATE \ifthen{$i=0$}{$a_0.\level \gets \min(a_0.\level+1,2m)$}
 \STATE \ifthen{$i=1$}{$a_1.\done \gets \tr$}
 \ENDIF
\IF{$a_0.\done = \tr \wedge a_1.\done = \tr \wedge \exists i \in \{0,1\}: a_i.\level < a_{1-i}.\level$}
 \STATE $a_i.\leader \gets \fl$
 \STATE $a_i.\level \gets a_{1-i}.\level$
\ENDIF
\end{algorithmic}
\end{algorithm}


\begin{lemma}
\label{lemma:game}
Consider the execution of $\quick$ under the uniformly random scheduler $\rs$ starting from a configuration
where at least one leader exists and $\done = 0$ and $\level = 0$ hold for any agent.
When all leaders finish deciding their levels (\ie $\done = 1$ holds for all leaders),
exactly one leader has the maximum level ($\max_{v \in V} v.\level$) with probability at least $1/16$.
\end{lemma}
\begin{proof}
Let $V'$ be the set of leaders at the initial configuration.
Let $X_v$ be the level computed by agent $v \in V'$.
We show that
$p_u=\Pr(X_u \ge \logn + 2 \wedge \bigwedge_{v \in V' \setminus \{u\}} X_v 
\leq \logn + 1) \geq 1/16n$
holds for any $u \in V'$. Then, the probability that some agent becomes the unique 
winner is obviously lower bounded by $\sum_{u \in V'} p_u \geq 1/16$.
Thus, the lemma holds because when $\done=1$ holds for all leaders,
every agent except for the unique winner has a smaller level or must have become a follower before. 
Since $\Pr(X_u \ge \logn + 2) > 1/8n$ holds,
it suffices to show $q = \Pr(\bigwedge_{v \in V' \setminus \{u\}} X_v \leq \logn +1 \mid X_u \ge \logn + 2) \geq 1/2$.
By the union bound, we have 
$q 
= 1 - \sum_{v \in V' \setminus \{u\}} \Pr(X_v \ge \logn + 2 \mid X_u \ge \logn+ 2)$.
When two agents $u$ and $v$ interact with each other, one of them necessarily reaches the decision of its level. That is, $u$ and $v$ have at most one common interaction until either one decides its level. This implies that to obtain $X_v \ge \logn  + 2$ under the condition $X_u \ge \logn +2$, $v$ must observe at least $\logn + 1$ heads at the coin flips independently of the first $\logn + 2$ coin flips by $u$. That is, we have $\Pr(X_v \ge \logn + 2 \mid X_u \ge \logn+ 2) \leq (1/2)^{\logn + 1} \le 1/2n$, and thus,
$q \geq  1 - n \cdot (1/2n) \geq 1/2$. 
\end{proof}

\section{Time-optimal Loosely-stabilizing Leader Election}
\label{sec:plog}

In this section, we give a loosely-stabilizing leader election
protocol $\popt(\para)$,
where the integer $\para \ge 1$ is a design parameter controlling the performance of the protocol. Starting from any initial configuration, this protocol
reaches a safe configuration within $O(\para n\log n)$ steps
and keeps the single leader in the following $\Omega(n^{\para})$ steps.
The number of states per agent is $\Theta(\para m)=\Theta(\para \log n)$.
In the rest of this paper,
we use terminologies ``\emph{with high probability}''
to mean ``with probability $1-O(1/n)$''
and ``\emph{with very high probability}''
to mean ``with probability $1-O(1/n^{\para})$''.
Further, the terminology ``quickly'' is used for 
implying ``within $O(\para n \log n)$ steps''.

\subsection{Protocol in a Nutshell}
\label{sec:overview}
The protocol $\popt(\para)$ elects a unique leader by iteratively performing the following two phases, 
both taking $\Theta(\para n \log n)$ steps with very high probability.
\begin{itemize}
 \item \textbf{Check phase}:
The protocol checks whether the population has
at least one leader. Each leader agent propagates a heartbeat message to all others using the epidemics. The agents not receiving that message until the end of the phase
conclude that the population has no leader, and they become leaders.
Since two or more agents may become leaders, they are filtered in
the election phase.
 \item \textbf{Election phase}:
Each agent performs QE. As shown in Lemma~\ref{lemma:game}, this phase 
decreases the number of leader agents to one with a constant probability.
\end{itemize}
There are two major issues for implementing these phases: how to realize
a loosely-stabilizing synchronization mechanism
to yield the transition between two phases, and how to combine it 
with the task of each phase using only a small number of states. The protocol CHVP, 
stated in Section~\ref{sec:countdown}, is one of the possible solutions for the first issue, which provides a loosely-stabilizing (synchronized) timeout mechanism; thus, it can be utilized for global phase synchronization.
However, addressing the second issue is, however, more challenging. Since
the check phase only consumes a constant number of states, it is easily combined 
with CHVP. In the election phase, both QE and CHVP internally 
keep a variable of a non-constant size. The former manages a variable $\level \in [0,2m]$, and the latter manages a variable whose range is $[0,O(\para m)]$, as 
we will see in Section \ref{sec:variables}. Thus, to bound the number of states 
by $O(\para m)=O(\para \log n)$ in total, they must share a single non-constant variable.

We resolve this matter by designing a new loosely-stabilizing task sharing scheme
called $\emph{mode switching}$.
Unlike the task-sharing techniques in the past literature~\cite{GSU18,SOI+19+}, 
it \emph{dynamically} changes the mode of each agent during the election phase. 
The two modes respectively correspond to synchronization and QE,
and each agent is engaged in the task associated with its own mode.
In total, 
the protocol is equipped with three different roles of agents, \ie check phase, synchronization in election phase, and QE in election phase. We call each role a \emph{class} of agents, and they are respectively 
referred to as \emph{checker}, \emph{synchronizer}, and \emph{elector}.
It should be noted that dynamic mode change is crucial for attaining loose 
stabilization: 
A non-correct initial configuration filled by electors obviously causes a 
deadlock because the timeout of the election phase never occurs forever.
Thus, it is indispensable to install a mechanism that changes electors to synchronizers. That mechanism, however, prevents the quick propagation of the maximum level in QE owing to the lack of a sufficiently large number of electors (recall that even agents not involved in the lottery game must work as a medium in the epidemic).
In fact, if only $o(n)$ electors remain, we cannot guarantee with very high probability that the epidemic of the maximum level finishes quickly. This observation implies that the mode change from synchronizers to electors is also necessary.

The remaining concern is how to design synchronization and QE with adapting to
dynamic change. The task of QE is robust for such
dynamics if an agent with mode change always joins as a follower. 
However, CHVP is not robust because the countdown timer is rewound by a newly joining 
agent with a high counter value. Fortunately, we can obtain an alternative solution 
for this matter: simply using a local countdown timer, which just counts the number of interactions performed
by the timer holder. While CHVP is necessary to recover global synchronization from the highly deviated situations where two agents are in different phases or have two counter values with a large difference, we can delegate such a role entirely to the check phase.
Then, the election phase can use the timeout mechanism not necessarily synchronized among all agents.

%


\subsection{Variables and Groups}
\label{sec:variables}

\begin{table}[t]
\caption{Variables used in protocol $\popt$}
\label{tbl:variables}
\center
\begin{tabular}{|l|l|l|}
 \hline
 & Variable name & Initial value\\ \hline 
 \multirow{3}{*}{Common variables}
 &  $\leader \in \{\fl,\tr\}$ & - \\ \cline{2-3}
 &  $\phase \in \{\checking,\elect\}$ & - \\ \cline{2-3}
 &   $\mode \in \{\cand,\timer\}$ & - \\ \cline{2-3} \hline

  \multirow{2}{*}{Variables for checkers}
 & $\rtimer \in [0,\rmax]$ & $\rmax$ \\ \cline{2-3}
 & $\detect \in \{\fl,\tr\}$ & $\leader$ \\ \hline

 \multirow{2}{*}{Variables for electors}
 & $\level \in [0,2m]$ & 0 \\ \cline{2-3}
 & $\done \in \{\fl,\tr\}$ & \fl \\ \hline

 Variables for synchronizers
 & $\btimer \in [0,\bmax]$ & $\bmax$\\ \hline

\end{tabular}
\end{table}

\begin{table}[t]
\caption{Descriptors for specific subsets of agents}
\label{tbl:notation}
\center
\begin{tabular}{l l}
\hline
$\vl= \{v \in V \mid v.\leader = \tr\}$ &$\vf= \{v \in V \mid v.\leader = \fl\}$ \\
$\vcheck= \{v \in V \mid v.\phase = \checking\}$ &$\velect= \{v \in V \mid v.\phase = \elect\}$ \\
$\vcand= \{v \in \velect \mid v.\mode = \cand\}$ &$\vtimer= \{v \in \velect \mid v.\mode = \timer\}$ \\
$\vclarge = \{v \in \vcheck \mid \rmid \le v.\rtimer \le \rmax\}$
& $\vcsmall = \{ v \in \vcheck \mid 0 \le v.\rtimer < \rmid\}$\\
$\vdone = \{v \in \vl \cap \vcand \mid v.\done = 1\}$
& $\vundone = \{v \in \vl \cap \vcand \mid v.\done = 0\}$\\
\hline
\end{tabular}
\end{table}

For describing the protocol, we use two (hard-coded) fixed values, $\rmax$ and $\bmax$, both of which are $\Theta(\para m)=\Theta(\para \log n)$ for sufficiently large hidden constants. We also define 
$\rmid = c \rmax$ for an appropriate $1 > c > 0$ such that $c / (1 - c)$ becomes sufficiently large.
All the hidden constants are appropriately fixed in the ``on-demand'' manner in the proof details. We also assume $n \geq 3$
for the simplicity of argument; however, it is not essential. 
It can be easily observed that this protocol is a self-stabilizing leader election protocol in the case of $n=2$.

The set of variables used in protocol $\popt$ is shown in Table~\ref{tbl:variables}.
As stated in Section~\ref{sec:overview}, in protocol $\popt$, there are three classes of agents: checkers, synchronizers, and electors. Each class has a set of variables specific for the associated task, and an agent manages the variables related to its own class as well as the set of common variables. Note that the list of variables in Table~\ref{tbl:variables} contains two $\Theta(\para \log n)$-state variables ($\rtimer$ and $\btimer$) and one $\Theta(\log n)$-state variable ($\level$), but they are used exclusively. That is, at any configuration, each agent has the responsibility of managing only one of the 
three. Thus, the total number of states necessary for storing all variables in Table~\ref{tbl:variables} is bounded by $O(\para \log n)$. The column ``Initial value'' in Table~\ref{tbl:variables} indicates the initial values set to 
class-specific variables. The initialization occurs when the agent changes its class.
For avoiding unnecessary complication, this initialization process is not explicitly stated in the pseudocode presented later.
The class of each agent is identified by two common variables 
$\phase$ and $\mode$. More precisely, the agent $v$ with $v.\phase = \checking$
is a checker, that with $v.\phase = \elect$ and $v.\mode = \cand$ an elector,
and that with $v.\phase = \elect$ and $v.\mode = \timer$ a synchronizer. The set of agents belonging to each class is denoted by $V_{\checking}$, $V_{\cand}$, and $V_{\timer}$ respectively. In addition, we introduce several notations for describing the set of agents satisfying some condition, as listed in Table~\ref{tbl:notation}.

\begin{algorithm}[t]
\caption{$\bfpopt$}
\label{al:popt}
\textbf{Interaction}  between initiator $a_0$ and responder $a_1$:
\hspace{-2cm}
\begin{algorithmic}[1]
\STATE \fordo{$i \in \{0,1\}: a_i \in \vtimer$}{$a_i.\leader \gets \fl$}
\IF{$a_0,a_1 \in \vcheck$}
\STATE $a_0.\detect \gets a_1.\detect \gets \max(a_0.\detect,a_1.\detect)$
\STATE $a_0.\rtimer \gets a_1.\rtimer \gets \max(a_0.\rtimer-1,a_1.\rtimer-1,0)$
\STATE \ifthen{$a_0.\rtimer = 0$}{$\toelect(0)$, $\toelect(1)$}
\ELSIF{
$\exists i \in \{0,1\}:a_{i} \in \velect \wedge a_{1-i} \in \vclarge$
}
\STATE $a_i.\phase \gets \checking$
\ELSIF{
$\exists i\in\{0,1\}:a_i \in \vcsmall \land a_{1-i} \in \velect$
}
\STATE $\toelect(i)$
\ENDIF
\IF{$a_0,a_1 \in \velect$}
\IF{$a_0,a_1 \in \vcand \cap \vf \wedge a_0.\level = a_1.\level$}
\STATE $a_1.\mode \gets \timer$
\ELSIF{$a_0,a_1 \in \vtimer$}
\STATE $a_i.\mode \gets \cand$ where
$i = \max \{j \in \{0,1\} \mid a_i.\btimer \ge a_{1-i}.\btimer\}$
\ENDIF
\STATE $\quick()$
 \STATE \ifthen{$a_0,a_1 \in \vdone \land a_0.\level=a_1.\level$}{$a_1.\leader \gets \fl$}
\STATE \fordo{$i \in \{0,1\}: a_i \in \vtimer$}{$a_i.\btimer \gets \max(a_i.\btimer-1,0)$}
\STATE \fordo{$i \in \{0,1\}: a_i.\btimer = 0$
}{
$a_i.\phase \gets \checking$
}
\ENDIF
\end{algorithmic}
\vblank

$\toelect(i)$
\begin{algorithmic}[1]
\setcounter{ALC@line}{20} 
\STATE $a_i.\phase \gets \elect$
\STATE \ifthen{$a_i.\detect = \fl$}{$(a_i.\leader,a_i.\mode) \gets (\tr,\cand)$}
\end{algorithmic}
\end{algorithm}

\subsection{Details of the Protocol}
\label{sec:sync}
The pseudocode of $\popt$ is shown in Algorithm \ref{al:popt}.
The main bodies of the two phases are realized by lines 3 and 17 and the procedure
$\toelect()$. Line 3, which corresponds to the check phase, performs the propagation 
of detect flags (\ie the existence of leader agents) using
the epidemic. Line 17 indeed corresponds to the task of QE.
The procedure $\toelect()$ corresponds to the phase transition from check to election, where the agent not detecting the existence of leaders becomes a leader. The remaining part is devoted to the synchronization mechanism including the mode switching scheme. Lines 4-5 correspond to the implementation of CHVP, where the timer variable $\rtimer$ is updated (line 4), and the transition to the election phase
is triggered when timeout occurs (line 5). Lines 6-9 are the mechanism 
supporting smooth phase transition, which is crucial for guaranteeing 
the correctness criteria of the synchronization mechanism explained later.
Lines 6-7 and 8-9 respectively address
the transition from check to election and its reversal.
Lines 11-20 correspond to the task for synchronizers and electors. The core of this part
is the mode switching scheme, described in lines 12-16. The switch from elector to synchronizer happens when a follower agent interacts with another follower with the same level (lines 12-13), and the opposite occurs when a synchronizer agent finds another synchronizer with a smaller timer value (lines 14-15). It is proved in
Section \ref{sec:analysis} that this scheme appropriately control the size of two classes.
Line 19 is the countdown of local timers held by synchronizers, and line 20 
is the phase transition from election to check. The leader elimination in line 18 is
not for the leader election itself, but rather to handle the initial configurations consisting only of leaders with the same level. Without this code, the protocol would be deadlocked in that case.
A synchronizer is always a follower, as guaranteed by line 1.

For stating the precise goal of the synchronization mechanism, we explain 
its intended behavior as well as the concise reason why such a behavior is attained.
\begin{enumerate}
 \item Starting from any configuration in $\call(\popt(\para))$, the population quickly reaches a configuration where $V = \vclarge$ holds with very high probability.
 The CHVP protocol shrinks the large deviation among all timers in the check phase.
 Therefore, once an agent goes back to the check phase from the election phase
 and resets its $\rtimer$ to $\rmax$, the population quickly reaches a configuration where $V=\vclarge$.
 One may think that the population gets stuck in the election phase once it reaches a configuration where all agents are electors (\ie there is no synchronizers). However, even starting from such a configuration, the population quickly creates at least one synchronizer because the following events occur with very high probability:
 (i) all leaders quickly decide their levels;
 (ii) the maximum level quickly propagates to the whole population as long as there is no synchronizer;
 (iii) since all the agents have the same level, the number of followers quickly becomes $\Theta(n)$ as long as there is no synchronizer; and
 (iv) two followers with the same level have an interaction quickly and one of them becomes a synchronizer.
 Once a synchronizer is created, some agent quickly goes back to the check phase because each synchronizer
 simply counts down its local timer.
 \item Once the current configuration satisfies $V = \vclarge$, CHVP decreases timer values (\ie $\rtimer$) with maintaining
a relatively smaller deviation among agents. 
 Since timer values of agents in $\vclarge$ are all $\Theta(\para \log n)$, the 
 check phase continues during $\Theta(\para n \log n)$ steps with very high probability. When an agent is timed out, it moves to the election phase. Then, owing to the low deviation of CHVP timers, no agent is still in $\vclarge$, and thus, the transition in lines 
 8-9 quickly takes all other agents to the election phase with very high probability.
 During this period, no agent goes back to the check phase from the election phase with very high probability because the upper limit of $\btimer$ is $\Theta(\para \log n)$ with a sufficiently large hidden constant.
\item In the election phase, the fastest timer (\ie the agent with the 
smallest timer value) of all synchronizers dominates the pace. 
Since it is never rewound, the election phase keeps $\Theta(\para n \log n)$ steps
with very high probability. Similar to the behavior from check to 
election, when an agent becomes a checker, all other agents are quickly brought back 
to the check phase with very high probability.
During this period, no agent goes to the election phase from the check phase with very high probability
because the upper limit of $\rtimer$ is $\Theta(\para \log n)$ with a sufficiently large hidden constant.
\end{enumerate}
The correctness criteria of the synchronization mechanism is that the system iterates
behaviors 2 and 3 with very high probability after recovery from unintended situations (by behavior 1), which is necessary
for our protocol to elect a unique leader in the loosely-stabilizing manner.
The formal proof of the correctness is presented in the next section. 

In $\quick()$, we expect that the largest level is quickly propagated to all leaders with very high probability.
Sudo \etal~\cite{SOI+19+} proved that this is true
if $|\vcand|=\Theta(n)$ holds
and $\vcand$ remains the same during this period.
However, $\popt$ frequently executes the mode switching from $\cand$ to $\timer$ and from $\timer$ to $\cand$.
Without the mode switching,
the number of agents with $\level=\smax$ is monotonically
non-decreasing, while
with the mode switching,
it decreases when an agent with $\level=\smax$
changes its mode from $\cand$ to $\timer$.
Therefore, we must evaluate
the effect of the mode switching on the speed of the propagation.
Fortunately,
there is no severe effect
of the mode switching for our purpose:
every leader in $\vdone$ whose $\level$ is not the largest
becomes a follower within $O(n \log n)$ steps
with probability $1-o(1)$, as we will prove in the next section (Lemma \ref{lemma:propWithSwitch}).


\section{Correctness}
\label{sec:analysis}
In this section,
we prove that
$\popt(\para)$ is an $(O(\para n \log n),\Omega(n^{\para+1}))$-loosely-stabilizing leader election protocol as Theorem \ref{theorem:main}.
In terms of parallel time (\ie the number of steps divided by $n$),
the convergence and the holding time of this protocol
is $O(\para \log n)$ and $\Omega(n^{\para})$, respectively.
Specifically, we prove
$\max_{C \in \call(\popt(\para))} \ect{\popt(\para)}{C}{\safe}= O(\para n \log n)$
and $\min_{C \in \safe} \allowbreak \eht{\popt(\para)}{C}{\sle}= \Omega(n^{\para+1})$,
where $\safe$ is the set of configurations defined as follows.
\begin{definition}[Safe configurations]
Define $\safe$ as the set of all configurations
where $V=\vclarge$ holds, exactly one leader $\lagent$ exists
in the population, and $\lagent.\detect = \tr$ holds.
\end{definition}

In the following, we first prove the correctness of synchronization in Section \ref{sec:sync_correct}.
Next, using the correctness of synchronization, we analyze the expected holding time 
and expected convergence time in Sections \ref{sec:holding} and \ref{sec:convergence}, respectively.

In the rest of this section, for any set $\mathcal{C}$,
we say that an execution \emph{enters} $\mathcal{C}$
when it reaches a configuration in $\mathcal{C}$.

\subsection{Synchronization}
\label{sec:sync_correct}

To express the claims in a formal manner, 
we first define the following notations.
\begin{itemize}
 \item $\cala_X$: the set of all configurations $\call(\popt(\para))$ where $V_{X}=V$ holds. For example, $\ccheck$ is the set of all configurations
where every agent is in the check phase (\ie $\vcheck = V$).
 \item $\celarge$: the set of all configurations in $\celect$
where $v.\btimer \ge \bmax/2$ holds for every $v \in \vtimer$.
 \item $\creset$:
the set of all configurations in $\call(\popt(\para))$ where
there is at least one agent $v \in \vcheck$ such that $v.\rtimer = \rmax$.
\end{itemize} 
The goal of this subsection is to prove the following three lemmas
(Lemmas \ref{lemma:cclarge}, \ref{lemma:goElection}, \ref{lemma:goChecking}).
Intuitively, Lemma \ref{lemma:cclarge} claims that synchronization is recovered quickly with very high probability
from any configuration in $\call(\popt(\para))$, and 
Lemmas \ref{lemma:goElection} and \ref{lemma:goChecking} claims that
once the synchronization is recovered, the check phase and the election phase are iterated thereafter,
both taking $\Theta(\para n \log n)$ steps with sufficiently large hidden constants, with very high probability.

\begin{lemma}
\label{lemma:cclarge}
Let $C_0$ be any configuration in $\call(\popt(\para))$
and let $\Xi = \Xi_{\popt(\para)}(C_0,\rs)$.
Execution $\Xi$ enters $\cclarge$
quickly with very high probability.
\end{lemma}

\begin{lemma}
\label{lemma:goElection}
Let $C_0$ be any configuration in $\cclarge$
and let $\Xi = \Xi_{\popt(\para)}(C_0,\rs)$.
Then, the following hold with very high probability:
\begin{enumerate}
    \item execution $\Xi$ enters $\celarge$ quickly,
    \item no agent moves from the election phase to the check phase
    before $\Xi$ enters $\celarge$, and
    \item execution $\Xi$ stays in $\ccheck$ for $\Omega(n \rmid)=\Omega(\para n \log n)$ steps.
\end{enumerate}
\end{lemma}

\begin{lemma}
\label{lemma:goChecking}
Let $C_0$ be any configuration in $\celarge$
and let $\Xi = \Xi_{\popt(\para)}(C_0,\rs)$.
Then, the following hold with very high probability:
\begin{enumerate}
    \item execution $\Xi$ enters $\cclarge$ quickly,
    \item no agent moves from the check phase to the election phase
    before $\Xi$ enters $\cclarge$, and
    \item execution $\Xi$ stays in $\celect$ for $\Omega(n \rmid)=\Omega(\para n \log n)$ steps.
\end{enumerate}
\end{lemma}

In what follows, we first prove Lemma \ref{lemma:cclarge}
by giving three supplemental lemmas (Lemmas \ref{lemma:done}, \ref{lemma:fromDone}, \ref{lemma:fromTimer}, and  \ref{lemma:creset}).
We next prove Lemmas \ref{lemma:goElection} and \ref{lemma:goChecking}.





\begin{lemma}
\label{lemma:done}
Starting from any configuration $C_0 \in \celect$,
execution $\Xi_{\popt(\para)}(C_0,\rs)$
quickly enters $\creset$
or reaches a configuration in $\celect$ satisfying
$\vundone = \emptyset$
with very high probability.
\end{lemma}
\begin{proof}
Let $\Xi=\Xi_{\popt(\para)}(C_0,\rs)$.
When an agent goes back to the check phase from the election phase, it substitutes $\rmax$ for its $\rtimer$.
Therefore, $\Xi$ never leaves $\celect$ until it enters $\creset$.
Let $v$ be any agent that satisfies $v \in \vundone$ in $C_0$.
As long as $v \in \vundone$ holds, 
$v$ makes a coin flip every time $v$ has an interaction.
Since every agent joins an interaction with probability $2/n$ at each step,
by the Chernoff bound
\footnote{
We often use the Chernoff bound in this paper.
We quote the Chernoff bound in Appendix
for the readers who are not familiar with probability.
},
$v$ has $2 \lg n$ or more interactions
within sufficiently large $O(n \log n)$ steps
with probability $1-O(1/n^{2})$.
Therefore, $v.\done=1$ holds within $O(n \log n)$ steps with
probability $1-(1/2)^{2 \log n}-O(1/n^{2})=1-O(1/n^{2})$ because each coin flip results in ``tail'' with probability exactly $1/2$,
by which $v$ leaves $\vundone$. By the union bound, execution $\Xi$ enters $\creset$ or reaches a configuration in $\celect$ satisfying $\vundone = \emptyset$ within $O(n \log n)$ steps with probability $1-O(1/n)$.
We obtain the lemma by repeating this analysis $\para$ times.
\end{proof}

\begin{lemma}
\label{lemma:fromDone}
Starting from any configuration $C_0 \in \celect$
satisfying $\vundone = \emptyset$,
execution $\Xi_{\popt(\para)}(C_0,\rs)$
quickly reaches a configuration in $\celect$
satisfying $\vtimer \neq \emptyset$ with very high probability.
\end{lemma}
\begin{proof} 
Let $\Xi = \Xi_{\popt(\para)}(C_0,\rs)$.
In execution $\Xi$,
no leader has $\done = \fl$
before $\Xi$ reaches a configuration in $\celect$ where $\vtimer \neq \emptyset$. Therefore, it suffices to show that
$\Xi$ reaches a configuration in $\celect$ where $\vtimer \neq \emptyset$ within $O(n \log n)$ steps with high probability
because we obtain the lemma by repeating this trial $\para$ times.

While both $\vundone = \emptyset$ and $\vtimer = \emptyset$ hold,
nothing prevents the epidemic from propagating the maximum value of $\level$s.
Thus, by Lemma \ref{lemma:epidemic},
the maximum value is propagated to the whole population
within $O(n \log n)$ steps with high probability \cite{fast}.
Once all agents have the same level, 
the number of followers increases by one every time two leaders meet. 
As long as $|\vf| < n/2$ and $\vtimer=\emptyset$ hold,
two leaders meet each other with probability at least $1/4$ at each step.
Hence, $|\vf| \ge n/2$ or $\vtimer \neq \emptyset$ holds
within $O(n \log n)$ steps with high probability.
In the former case, at each step thereafter,
two followers have an interaction
and one of them becomes a synchronizer (at Line 13)
with a constant probability.
Therefore, $|\vtimer| \neq \emptyset$ holds 
within $O(\log n)$ steps
with high probability.
\end{proof}

\begin{lemma}
\label{lemma:fromTimer}
Starting from any configuration $C_0 \in \celect$
satisfying $\vtimer \neq \emptyset$ holds,
execution $\Xi_{\popt(\para)}(C_0,\rs)$
quickly enters $\creset$
with very high probability.
\end{lemma}

\begin{proof}
Before $\Xi_{\popt(\para)}(C_0,\rs)$ enters $\creset$,
the smallest $\btimer$ in the population
(\ie $\min_{v \in \vtimer} \btimer$)
is monotonically non-increasing.
It decreases by one or a timeout of $\btimer$ occurs when a synchronizer with the smallest $\btimer$ has an interaction,
which occurs with probability at least $2/n$ at each step.
Since $\bmax = \Theta(\para \log n)$,
by the Chernoff bound,
some synchronizer encounters the timeout of $\btimer$ quickly with very high probability.
\end{proof}


\begin{lemma}
\label{lemma:creset}
Starting from any configuration $C_0 \in \call(\popt(\para))$,
execution $\Xi_{\popt(\para)}(C_0,\rs)$
quickly enters $\creset$
with very high probability.
\end{lemma}
\begin{proof}
Let $\Xi=\Xi_{\popt(\para)}(C_0,\rs)$.
Note that $\Xi$ enters $\creset$ whenever
an agent goes back to the check phase from the election phase.
Before $\Xi$ enters $\creset$ or $\celect$, 
$\max_{v \in \vcheck} v.\rtimer$ is monotonically non-increasing and this value decreases at a pace faster than or equal to the pace at which the maximum value of variable $y$ decreases in the CHVP protocol in Section \ref{lemma:countdown}.
Therefore, by Lemma \ref{lemma:countdown} with $l_1=\rmax$, $l_2=0$, $k=\para$, $\Xi$ reaches
a configuration $C' \in \creset \cup \celect$
quickly with very high probability.
Once $\Xi$ enters $\celect$, it enters in $\creset$
quickly with very high probability
by Lemmas \ref{lemma:done}, \ref{lemma:fromDone},
and \ref{lemma:fromTimer}.
\end{proof}



\begin{proof}[Proof of Lemma \ref{lemma:cclarge}]
By Lemma \ref{lemma:creset}, we can assume that $C_0 \in \creset$.
Since we assume $\rmax-\rmid= O(\para \log n)$ with a sufficiently large hidden constant, by Lemma \ref{lemma:prop} with $l=\rmax$ and $k=\para$,
execution $\Xi_{\popt(\para)}(C_0,\rs)$ enters $\cclarge$ quickly with very high probability.
\end{proof}

\begin{proof}[Proof of Lemma \ref{lemma:goElection}]
The last claim is trivial because
each agent has an interaction with probability $2/n$ at each step
and at least one agent must have $\rmid$ interactions
until execution $\Xi$ leaves $\ccheck$.
We prove the first and the second claims below.

By Lemma \ref{lemma:countdown} with $k=\para$,
$\Xi$ enters $\ccsmall$
within $O(n(\rmax - \rmid))=O(\para n \log n)$ steps
or at least one agent goes to the electing phase during the period
with very high probability.
Since we assume that $\rmid/(\rmax-\rmid)$ is a sufficiently large constant,
together with the third claim and the union bound,
we observe that $\Xi$ enters $\ccsmall$ quickly with very high probability.
Thereafter, by Lemma \ref{lemma:countdown} with $k=\para$,
at least one agent goes to the electing phase 
within $O(\para n \log n)$ steps with very high probability. 
Remember that whenever an agent in $\vcsmall$
and an agent in $\velect$ meet,
the former moves to the electing phase. 
Hence, once an agent goes to the electing phase,
the agents in the population go to the electing phase one after another
in completely the same way as the epidemic protocol in Section \ref{sec:epidemic}.
Therefore, by Lemma \ref{lemma:epidemic} with $k=\para$,
$\Xi$ reaches a configuration $C \in \celect$
quickly with very high probability.
Since $\bmax$ is sufficiently large,
no agent has more than $\bmax/2$ interactions during this period
with very high probability.
Therefore, $C \in \celarge$ holds with very high probability.
From the above, we conclude that the first and second claims also hold.
\end{proof}

\begin{proof}[Proof of Lemma \ref{lemma:goChecking}]
The last claim is trivial because
each agent has an interaction with probability $2/n$ at each step
and at least one agent must have $\bmax/2$ interactions
until execution $\Xi$ leaves $\celect$.
By Lemmas \ref{lemma:done}, \ref{lemma:fromDone}, and \ref{lemma:fromTimer},
$\Xi$ reaches a configuration $\creset$ within $O(\para n \log n)$ steps with very high probability.
Thereafter, in completely the same way as given in the second paragraph of the proof of Lemma \ref{lemma:goElection},
we can prove that $\Xi$ enters $\cclarge$ quickly (by the epidemic) and
$\vcsmall=\emptyset$ always holds during this period with very high probability.
Thus, the first claim holds.
No agent goes to the election phase
when it belongs to $\vclarge$, from which the second claim follows.
\end{proof}

\subsection{Holding Time}
\label{sec:holding}

\begin{lemma}
\label{lemma:holding} 
$\min_{C \in \safe} \eht{\popt(\para)}{C}{\sle}= \Omega(n^{\para+1})$.
\end{lemma}

\begin{proof}
Let $C_0$ be any configuration in $\safe$
and let $\Xi=\Xi_{\popt(\para)}(C_0,\rs)$.
Since the unique leader $\lagent$ in $C_0$ satisfies $\lagent.\detect = 1$ and $\rmid$ is a sufficiently large $\Theta(\para \log n)$ value,
by Lemma \ref{lemma:epidemic} with $k=\para$ and the third claim of Lemma \ref{lemma:goElection},
all agents detect the existence of a leader quickly
with very high probability.
Therefore, by the first and second claims of Lemma \ref{lemma:goElection},
it holds with very high probability that
$\Xi$ quickly reaches a configuration $C' \in \celarge$ where $\lagent$ is the unique leader
and no agent has a higher level than $\lagent.\level$.
This is because $\lagent$ goes to the election phase exactly once before $\Xi$ reaches $C'$ with very high probability,
and the level of every agent is initialized to zero when it goes to the election phase. 
Thereafter,
$\lagent$ never becomes a follower before it goes back to the check phase and goes to the election phase again;
this is because only a leader can increase $\max_{v \in V} v.\level$.
By Lemma \ref{lemma:goChecking}, 
$\Xi$ quickly enters $\cclarge$ again 
and $\lagent$ does not move to the election phase from the check phase during this period with very high probability. 
At this time, from the above discussion, $\lagent$ is still
the unique leader in the population and $\lagent.\detect=1$ holds.
This means that the population has come back to $\safe$.

Now, we observed that an execution of $\popt$ under the uniformly random scheduler $\rs$ starting from any configuration in $\safe$
goes back to a configuration in $\safe$ after $\Theta(\para n \log n)$ steps and $\lagent$ is always the unique leader during this period
with very high probability. Therefore, letting $X=\min_{C \in \safe} \eht{\popt(\para)}{C}{\sle}$, we have $X \ge (1-O(n^{-\para}))(\Theta(\para n \log n)+X)$. Solving this inequality gives $X=\Omega(\para n^{\para+1}\log n)= \Omega(n^{\para+1})$.
\end{proof}

\subsection{Convergence Time}
\label{sec:convergence}

\begin{lemma}
\label{lemma:propWithSwitch}
Let $C_0$ be a configuration in $\celect$ where $\vundone = \emptyset$
holds and let $\Xi=\Xi_{\popt(\para)}(C_0,\rs)$.
Let $V' \subset \vl \cap \vcand$ be the set of leaders
whose $\level$ are not the largest in $C_0$.
Then, execution $\Xi$ reaches a configuration in $\celect$ where 
all agents in $V'$ are followers or enters $\creset$ within $O(n \log n)$ steps
with probability $1-o(1)$.
\end{lemma}

\begin{proof}
In this proof, we ignore the case that some agent goes back to the check phase (\ie $\Xi$ enters $\creset$)
because this ignorance
only decreases the probability claimed in the lemma.
Let $l$ be the maximum level of the population
(\ie $\max_{v \in \vcand} v.\level$) in $C_0$.
To obtain the lemma, 
it suffices to show that
all leaders in $V'$ observe
the maximum level $l$
and become followers within $O(n \log n)$ steps
with probability $1-o(1)$.

First, we analyze $\nva = |\vcand|$ in execution $\Xi$.
This value increases by one if two agents in $\vtimer$ meet,
and it decreases by one
if two followers with the same level in $\vcand$ meet. 
Therefore, at each step where $\nva \le n/3$,
$\nva$ increases with probability at least $4/9$,
while $\nva$ decreases with probability at most $1/9$.
The gap of these probabilities and the Chernoff bound guarantee that
even if $\nva < n/3$ in $C_0$, 
$\nva$ reaches $n/3$ within $O(n)$ steps with high probability.
Let $C'$ be the configuration at this time.
Once $\Xi$ reaches $C'$,
the above gap of probabilities, the Chernoff bound, 
and the union bound guarantee that
$\nva \ge n/4$ always holds for arbitrarily
large $\Omega(n \log n)$ steps with high probability.
Thus, we can assume $\nva \ge n/4$ in the following discussion
on execution $\Xi$ after $C'$.

Consider the postfix of $\Xi$ after $C'$.
Let $\nvm = |\{v \in \vcand \mid v.\level = l\}|$.
This value increases by one if
an interaction happens between two agents in $\vcand$ such that
one has the maximum level $l$
and the other has a lower level.
It decreases by one if an interaction happens between
two followers in $\vcand$, both with level $l$.
Note that $\nvm \ge 1$ always holds
because $\nvm$ decreases only if two agents with level $l$
have an interaction.
Since we assume $\nva \ge n/4$,
at each step where $\nvm \le n/8$,
$\nvm$ increases with probability $\pinc \ge (\nvm \cdot (n/4-\nvm))/\binom{n}{2} \ge \nvm/(4n)$,
while $\nvm$ decreases with probability $\pdec \le \nvm^2/n^2$.
As long as $n/2^9 \le \nvm \le n/2^8$,
we have $\pinc \ge 1/2^{11}$ and $\pdec \le 1/2^{16}$.
This large difference between $\pinc$ and $\pdec$ guarantees that
once $\nvm$ reaches $n/2^8$,
$\nvm \ge n/ 2^9$ always holds for arbitrarily large $\Omega(n \log n)$ steps with high probability,
by the Chernoff bound and the union bound.
During this period, each leader in $V'$
meets an agent in $\vcand$ with the maximum level $l$
with probability $\Omega(1/n)$ at each step.
Thus, once $\nva \ge n/2^8$ holds,
all leaders in $V'$ become followers
within $O(n \log n)$ steps with high probability.

Thus, all we have to do is to show that $\nvm \ge n/2^8$ holds
within $O(n \log n)$ steps starting from $C'$.
First, we show that $\nvm$ reaches $24\ln n$ or larger
within $O(n \log n)$ steps starting from $C'$.
When $\nvm < 24 \ln n$,
$\pinc = \Omega(1/n)$ and $\pdec = O((\log^2 n) / n^2)$ always hold.
Therefore, by the Chernoff bound, $\nvm$ reaches $24\ln n$ or larger
within $O(n\log n)$ steps with high probability.
Next, for any integer $k$ such that $24\ln n \le k < n/2^8$,
we show that once $\nvm$ reaches $k$,
$\nvm$ reaches $2k$ with high probability.
As long as $k/2 \le \nvm \le 2k$,
we have $\pinc \ge k/8n$ and $\pdec \le 4k^2/n^2 < k/64n$.
Therefore, by the Chernoff bound, we have the followings:
\begin{itemize}
 \item 
during the first $16n$ steps,
$\nvm$ is always $k/2$ or larger with probability 
at least $1-e^{-(1/3)\cdot (k/4)}=1-O(1/n^2)$,
 \item
during the first $x \ge 16n$ steps,
$\nvm$ increases at least $xk/16n$ times
with probability $1-O(1/n^2)$, and 
 \item
during the first $x \ge 16n$ steps,
$\nvm$ decreases at most $xk/32n$ times
with probability $1-O(1/n^2)$.
\end{itemize}
Therefore, by the union bound (for $x=16n,16n+1,\dots,32n$),
$\nvm$ reaches $k+(2k-k)=2k$ within $32n$ steps
with high probability.
Therefore, once $\nvm$ reaches $24 \ln n$ or larger value,
it doubles in every $32n$ steps with high probability
until it reaches $n/2^8$.
Thus, $\nvm$ reaches $n/2^8$ within $O(n \log n)$ steps
with probability $1-O((\log n)/n)=1-o(1)$.
\end{proof}

%
%
%

\begin{lemma}
\label{lemma:convergence} 
$\max_{C \in \call(\popt(\para))} \ect{\popt(\para)}{C}{\safe}= O(\para n \log n)$.
\end{lemma}

\begin{proof}
Let $C_0$ be any configuration in $\call(\popt(\para))$
and let $\Xi=\Xi_{\popt(\para)}(C_0,\rs)$.
It suffices to show that $\Xi$ enters $\safe$
quickly with a constant probability;
this is because, letting
$Y= \max_{C \in \call(\popt(\para))} \ect{\popt(\para)}{C}{\safe}$,
it yields $Y \le O(\para n \log n) + (1-\Omega(1)) Y$,
and this inequality gives $Y = O(\para n \log n)$.

We can assume $C_0 \in \celarge$ by Lemmas \ref{lemma:cclarge}
and \ref{lemma:goElection}.
By Lemma \ref{lemma:goChecking}, $\Xi$ reaches a configuration
$C' \in \cclarge$ within $O(\para n \log n)$ steps
and no agent goes to the electing phase from the check phase
during this period with very high probability.
An agent executes $\detect \gets \leader$ when it goes back to
the check phase (See Table \ref{tbl:variables}).
Therefore, we have $\exists v \in V: v.\detect = 1 \Leftrightarrow \vl \neq \emptyset$
in $C'$.
Hence, after $\Xi$ reaches $C'$, at least one follower becomes a leader when it goes to the electing phase if there exists no leader in $C'$.
Moreover, 
the agents initialize their $\level$ and $\detect$ to $0$
when they move to the election phase.
Therefore, by Lemmas \ref{lemma:goElection}, \ref{lemma:goChecking}, \ref{lemma:done}, and \ref{lemma:game},
$\Xi$ quickly reaches a configuration $C'' \in \celarge$
where exactly one leader, say $\lagent$, has the maximum level
with a constant probability.
Thereafter, by Lemma \ref{lemma:propWithSwitch},
$\Xi$ reaches a configuration in $\celect$ where only $\lagent$
is a leader within $O(n \log n)$ steps with probability $1-o(1)$.
In the next $O(\para n \log n)$ steps,
$\Xi$ enters $\safe$ with very high probability
by Lemma \ref{lemma:goChecking}.
By summing up all error probabilities,
we conclude that $\Xi$ enters $\safe$
within $O(\para n \log n)$ steps with a constant probability.
\end{proof}

\begin{theorem}
\label{theorem:main}
For any positive integer $\para$,
$\popt(\para)$ is an $(O(\para n \log n),\Omega(n^{\para}))$-loosely-stabilizing leader election protocol.
\end{theorem}
\begin{proof}
Immediately follows
from Lemmas \ref{lemma:holding} and \ref{lemma:convergence}.
\end{proof}


\section{Conclusion}
We gave a time-optimal loosely-stabilizing leader election protocol in the population protocol model. 
Let $n$ be the number of agents in the population.
Given a design parameter $\para \ge 1$
and integer $N \ge n$ such that $N$ is at most polynomial in $n$, the proposed protocol elects the unique leader
within $O(\para \log n)$ parallel time
starting from  any configuration
and keeps it for $\Omega(n^{\para})$ parallel time,
both in expectation. 



\bibliographystyle{abbrv}
\bibliography{paper6}

\begin{thebibliography}{10}

\bibitem{AAE+17}
D.~Alistarh, J.~Aspnes, D.~Eisenstat, R.~Gelashvili, and R.~L. Rivest.
\newblock Time-space trade-offs in population protocols.
\newblock In {\em Proceedings of the Twenty-Eighth Annual ACM-SIAM Symposium on
  Discrete Algorithms}, pages 2560--2579. SIAM, 2017.

\bibitem{AAG18}
D.~Alistarh, J.~Aspnes, and R.~Gelashvili.
\newblock Space-optimal majority in population protocols.
\newblock In {\em Proceedings of the Twenty-Ninth Annual ACM-SIAM Symposium on
  Discrete Algorithms}, pages 2221--2239. SIAM, 2018.

\bibitem{ADK+17}
D.~Alistarh, B.~Dudek, A.~Kosowski, D.~Soloveichik, and P.~Uzna{\'n}ski.
\newblock Robust detection in leak-prone population protocols.
\newblock In {\em International Conference on DNA-Based Computers}, pages
  155--171. Springer, 2017.

\bibitem{AG15}
D.~Alistarh and R.~Gelashvili.
\newblock Polylogarithmic-time leader election in population protocols.
\newblock In {\em Proceedings of the 42nd International Colloquium on Automata,
  Languages, and Programming}, pages 479--491, 2015.

\bibitem{original}
D.~Angluin, J.~Aspnes, Z.~Diamadi, M.~J. Fischer, and R.~Peralta.
\newblock Computation in networks of passively mobile finite-state sensors.
\newblock {\em Distributed Computing}, 18(4):235--253, 2006.

\bibitem{fast}
D.~Angluin, J.~Aspnes, and D.~Eisenstat.
\newblock Fast computation by population protocols with a leader.
\newblock {\em Distributed Computing}, 21(3):183--199, 2008.

\bibitem{jikoantei}
D.~Angluin, J.~Aspnes, M.~J. Fischer, and H.~Jiang.
\newblock Self-stabilizing population protocols.
\newblock {\em ACM Transactions on Autonomous and Adaptive Systems}, 3(4):13,
  2008.

\bibitem{BGK20}
P.~Berenbrink, G.~Giakkoupis, and P.~Kling.
\newblock Optimal time and space leader election in population protocols.
\newblock In {\em STOC 2020: 52nd Annual ACM Symposium on Theory of Computing},
  2020.

\bibitem{BKKO18}
P.~Berenbrink, D.~Kaaser, P.~Kling, and L.~Otterbach.
\newblock {Simple and Efficient Leader Election}.
\newblock In {\em 1st Symposium on Simplicity in Algorithms (SOSA 2018)},
  volume~61, pages 9:1--9:11, 2018.

\bibitem{BCER17}
A.~Bilke, C.~Cooper, R.~Els\"{a}sser, and T.~Radzik.
\newblock Brief announcement: Population protocols for leader election and
  exact majority with {$O(\log^2 n)$} states and {$O(\log^2 n)$} convergence
  time.
\newblock In {\em Proceedings of the 38th ACM Symposium on Principles of
  Distributed Computing}, pages 451--453, 2017.

\bibitem{Bur+19}
J.~Burman, D.~Doty, T.~Nowak, E.~E. Severson, and C.~Xu.
\newblock Efficient self-stabilizing leader election in population protocols.
\newblock {\em arXiv preprint arXiv:1907.06068}, 2019.

\bibitem{SIW12}
S.~Cai, T.~Izumi, and K.~Wada.
\newblock How to prove impossibility under global fairness: On space complexity
  of self-stabilizing leader election on a population protocol model.
\newblock {\em Theory of Computing Systems}, 50(3):433--445, 2012.

\bibitem{DS18}
D.~Doty and D.~Soloveichik.
\newblock Stable leader election in population protocols requires linear time.
\newblock {\em Distributed Computing}, 31(4):257--271, 2018.

\bibitem{GS18}
L.~G{\k{a}}sieniec and G.~Stachowiak.
\newblock Fast space optimal leader election in population protocols.
\newblock In {\em Proceedings of the Twenty-Ninth Annual ACM-SIAM Symposium on
  Discrete Algorithms}, pages 2653--2667. SIAM, 2018.

\bibitem{GSU18}
L.~G{\k{a}}sieniec, G.~Stachowiak, and P.~Uznanski.
\newblock Almost logarithmic-time space optimal leader election in population
  protocols.
\newblock In {\em The 31st ACM on Symposium on Parallelism in Algorithms and
  Architectures}, pages 93--102. ACM, 2019.

\bibitem{kakai}
T.~Izumi.
\newblock On space and time complexity of loosely-stabilizing leader election.
\newblock In {\em International Colloquium on Structural Information and
  Communication Complexity}, pages 299--312, 2015.

\bibitem{kyoukasyo}
M.~Mitzenmacher and E.~Upfal.
\newblock {\em Probability and Computing: Randomized Algorithms and
  Probabilistic Analysis}.
\newblock Cambridge University Press, 2005.

\bibitem{SM19}
Y.~Sudo and T.~Masuzawa.
\newblock Leader election requires logarithmic time in population protocols.
\newblock {\em arXiv preprint arXiv:1906.11121}, 2019.

\bibitem{kanjiko}
Y.~Sudo, J.~Nakamura, Y.~Yamauchi, F.~Ooshita, H.~Kakugawa, and T.~Masuzawa.
\newblock Loosely-stabilizing leader election in a population protocol model.
\newblock {\em Theoretical Computer Science}, 444:100--112, 2012.

\bibitem{SOI+18}
Y.~Sudo, F.~Ooshita, T.~Izumi, H.~Kakugawa, and T.~Masuzawa.
\newblock Logarithmic expected-time leader election in population protocol
  model.
\newblock {\em arXiv preprint arXiv:1812.11309}, 2018.

\bibitem{SOI+19+}
Y.~Sudo, F.~Ooshita, T.~Izumi, H.~Kakugawa, and T.~Masuzawa.
\newblock Logarithmic expected-time leader election in population protocol
  model.
\newblock In {\em Proceedings of the 21st International Symposium on
  Stabilizing, Safety, and Security of Distributed Systems}, pages 323--337,
  2019.

\bibitem{without}
Y.~Sudo, F.~Ooshita, H.~Kakugawa, and T.~Masuzawa.
\newblock Loosely-stabilizing leader election on arbitrary graphs in population
  protocols without identifiers nor random numbers.
\newblock In {\em International Conference on Principles of Distributed
  Systems}, 2015.

\bibitem{SOK+18}
Y.~Sudo, F.~Ooshita, H.~Kakugawa, T.~Masuzawa, A.~K. Datta, and L.~L. Larmore.
\newblock Loosely-stabilizing leader election for arbitrary graphs in
  population protocol model.
\newblock {\em IEEE Transactions on Parallel and Distributed Systems},
  30(6):1359--1373, 2018.

\bibitem{Sud+20}
Y.~Sudo, F.~Ooshita, H.~Kakugawa, T.~Masuzawa, A.~K. Datta, and L.~L. Larmore.
\newblock Loosely-stabilizing leader election with polylogarithmic convergence
  time.
\newblock {\em Theoretical Computer Science}, 806:617--631, 2020.

\end{thebibliography}

\clearpage

\appendix

\section*{Appendix}
\section{Chernoff Bounds}
\begin{lemma}[\cite{kyoukasyo}, Theorems 4.4, 
4.5] 
\label{lemma:chernoff} 
Let $X_1,\dots,X_s$ be independent Poisson 
trials, 
and let $X = \sum_{i=1}^s X_i$. 
Then 
\begin{align} 
\label{eq:upperdouble} 
\forall \delta,~0 \le \delta \le 1:~ 
&
\Pr(X \ge (1+\delta)\ex[X]) \le 
e^{-\delta^2\ex[X]/3},\\ 
\label{eq:lowerhalf} 
\forall \delta,~0 < \delta < 1:~ 
&
\Pr(X \le (1-\delta)\ex[X]) \le 
e^{-\delta^2\ex[X]/2}. 
\end{align} 
\end{lemma} 
\end{document}